\documentclass[sigconf, usenames, dvipsnames]{acmart}

\usepackage{algorithm}
\usepackage{algorithmic}
\usepackage{multirow}
\usepackage{slashbox}
\usepackage{soul}
 \usepackage{relsize}
 \usepackage{colortbl}
 \definecolor{paleblue}{RGB}{194,207,237}

\newcommand{\Scal}{\mathcal{S}}
\newcommand{\Pro}{\mathrm{Pr}}
\newcommand{\F}{\mathcal{F}}

    \makeatletter
\def\@fnsymbol#1{\ensuremath{\ifcase#1\or \dagger\or \ddagger\or
   \mathsection\or \mathparagraph\or \|\or **\or \dagger\dagger
   \or \ddagger\ddagger \else\@ctrerr\fi}}
    \makeatother
\AtBeginDocument{%
  \providecommand\BibTeX{{%
    \normalfont B\kern-0.5em{\scshape i\kern-0.25em b}\kern-0.8em\TeX}}}

\copyrightyear{2020}
\acmYear{2020}
\setcopyright{acmlicensed}\acmConference[AFT '20]{2nd ACM Conference on Advances in Financial Technologies}{October 21--23, 2020}{New York, NY, USA}
\acmBooktitle{2nd ACM Conference on Advances in Financial Technologies (AFT '20), October 21--23, 2020, New York, NY, USA}
\acmPrice{15.00}
\acmDOI{10.1145/3419614.3423265}
\acmISBN{978-1-4503-8139-0/20/10}
\acmSubmissionID{72}

 \newcommand{\dcp}[1]{\textcolor{blue}{{\scriptsize{David:}}#1}}

\newcommand{\rithvik}[1]{\textcolor{cyan}{{\scriptsize{Rithvik: }}#1}}

\begin{document}

\title{Defending Against Malicious Reorgs in Tezos Proof-of-Stake}

\author{Michael Neuder}
\affiliation{%
  \institution{Harvard University}
}
\email{michaelneuder@g.harvard.edu}

\author{Daniel J. Moroz}
\affiliation{%
  \institution{Harvard University}
  }
\email{dmoroz@g.harvard.edu}

\author{Rithvik Rao}
\affiliation{%
  \institution{Harvard University}
  }
\email{rithvikrao@college.harvard.edu}

\author{David C. Parkes}
\affiliation{%
  \institution{Harvard University}
}
\email{parkes@eecs.harvard.edu}

% \renewcommand{\shortauthors}{Neuder, et al.}

\begin{abstract}
Blockchains are intended to be immutable, so an attacker who is able to delete  transactions through a chain reorganization (a {\em malicious reorg}) can  perform a profitable double-spend attack. We study the rate at which an attacker can execute reorgs in the Tezos Proof-of-Stake protocol. As an example, an attacker with 40\% of the staking power is able to execute a 20-block malicious reorg at an average rate of once per day, and the attack probability increases super-linearly as the staking power grows beyond 40\%. Moreover, an attacker of the Tezos protocol knows in advance when an attack opportunity will arise, and  can use this knowledge to arrange transactions to double-spend.  We show that in particular cases, the Tezos protocol can be adjusted to protect against deep reorgs. For instance, we demonstrate protocol parameters that reduce the rate of length-20 reorg opportunities for a 40\% attacker by two orders of magnitude. We also observe a trade-off between optimizing for robustness to deep reorgs (costly deviations that may be net profitable because they enable double-spends) and robustness to selfish mining (mining deviations that result in typically short reorgs that are profitable even without double-spends). That is, the parameters that optimally protect against one make the other attack easy. Finally, we develop a method that monitors the Tezos blockchain health with respect to malicious reorgs using only publicly available information.
\end{abstract}




%   \includegraphics[width=\textwidth]{sampleteaser) 

\maketitle

\section{Introduction}

Blockchains are designed to be immutable in order to protect against attackers who seek to delete transactions through chain reorganizations ({\em malicious reorgs}). Any attacker who causes a reorg of the chain could {\em double-spend} transactions, meaning they commit a transaction to the chain, receive some goods in exchange, and then delete the transaction, effectively robbing their counterparty. 
Nakamoto~\cite{nakamoto2008bitcoin} demonstrated that, in a {\em  Proof-of-Work} (PoW) setting, minority (<50\%) attacker forks have an exponentially decreasing probability of overtaking the honest chain (thus causing a reorg) as more honest blocks are built.
This ensures that transactions written to blocks will remain on the chain with high probability.
Because PoW protocols require significant energy expenditure and provide low transaction throughput, 
{\em Proof-of-Stake} (PoS) protocols are seen as a viable alternative and are being used by a number of projects (capitalizations as of September 2020):  Tezos \cite{goodman2014tezos} (\$1.8 billion), Cardano \cite{kiayias2017ouroboros} (\$2.5 billion), EOS \cite{eos} (\$2.6 billion), Nxt \cite{nxt} (\$11 million), and BlackCoin \cite{vasin2014blackcoin} (\$2.5 million).

A common feature of PoS protocols not present in PoW is {\em global predictability}. That is, participants in the consensus protocol ({\em stakers}) are able to determine far in advance  exactly when they will have the opportunity to mine blocks (in the case of Tezos, days ahead of time). This helps validators know when to create or validate a block, but this predictability also makes a double-spend attack easier to perform, as an attacker knows precisely when the opportunity to reorg will arise, allowing them to send a soon-to-be-deleted transaction to an unsuspecting counterparty at exactly the right time. This precision is not possible to achieve under PoW protocols, as no miner knows in advance when they will successfully mine a block.  

In this work, we analyze the {\em Tezos PoS protocol} and develop a method to calculate the rate of malicious reorgs.
We show, for example, that an attacker with 40\% of the total stake is able to execute a 20-block malicious reorg at an average rate of once per day, and that an attacker's power grows quickly as its staking percentage increases beyond that point. 
We also study the extent to which adjusting the design parameters of the protocol can protect against such attacks. We find a set of protocol parameters for Tezos PoS that  decrease the rate of attack opportunities by over 50\% compared to the current design choice for specific values of the attacker's stake. As an example, the 40\% attacker under the suggested protocol parameters only achieves a length-20 reorg once per year.

We also demonstrate a trade-off in setting these design parameters
between optimizing the PoS protocol for safety against deep reorgs versus safety against selfish mining. {\em Selfish mining} is defined as a mining strategy in which the attacker earns more protocol-prescribed rewards than behaving honestly by selectively withholding blocks instead of publishing them immediately. When feasible, selfish mining incentivizes dishonest behavior even without a double-spend (and double-spends are implausible in the typically shallow reorgs of a selfish-mine). Selfish mining has been extensively studied in PoW \cite{eyal2018majority,sapirshtein2016optimal,kwon2017selfish, nayak2016stubborn}.
We find that the parameters that optimally protect against deep reorgs also make selfish mining easier,  confirming the findings of Nomadic Labs \cite{analysisemmyplus}. Fundamentally, this is because opportunities for selfish mining occur more frequently in shorter reorgs (typically less than 5 blocks), while double-spend attacks require deeper reorgs (we consider reorgs up to length 80). In a selfish mining attack, the length of the attack is determined by the protocol rewards, and attacker continues as long as they earn more rewards than following the honest strategy. The length of a reorg needed to successfully execute a double-spend transaction depends on how many block confirmations a counterparty requires before releasing a good to the attacker. This varies depending on the blockchain, but is generally recommended to be 30 blocks for Tezos~\cite{30confs}.

We also develop a method that monitors the health of the Tezos blockchain with regard to reorg attacks, using only on-chain, publicly available information. This metric allows users to identify potentially vulnerable chain states and take extra care with accepting large transactions.

{\bf Outline.}
Section~\ref{sec:tezos-pos} presents the Tezos PoS model and describes our technique for computing the probability of a malicious reorg. 
Section~\ref{sec:results} presents the results of our analysis, showing frequency of attack opportunity as a function of attacker strength and attack depth.
For the purpose of this analysis, we estimate the probabilities of a malicious reorg using analytical methods when possible, but primarily with Monte Carlo and importance sampling (see Appendices \ref{app:mc} \& \ref{app:impsample} for more details). 
In Section~\ref{sec:protocol-mod}, we consider the relationship between parameters in the Tezos PoS protocol and the susceptibility of the protocol to reorgs and selfish mining. 
In Section \ref{sec:detection} we provide a method to monitor the likelihood of impending attacks based only on public information, and we demonstrate through simulation the effectiveness of this method.

\subsection{Related Work}
Deep reorgs are dangerous for any blockchain, and a number of studies have considered  vulnerability to these kinds of attacks.
\citet{nakamoto2008bitcoin} provided a simple analysis of the expected frequency with which a minority attacker could execute deep reorgs on Bitcoin and  \citet{powDoubleSpend} followed up in more detail, calculating the probability of different length reorgs given various attacker strengths. 

This question has also been investigated in PoS protocols. \citet{kiayias2017ouroboros} formally analyze Ouroboros, which underpins the Cardano blockchain, by studying the rate at which deep reorgs occur through a probabilistic analysis of states in which forks arise. They present security proofs as well as experimental results from their analysis, which demonstrate that long block confirmation wait times are required to avoid a double-spend from an adversary who has a large percentage of the stake. In another PoS analysis, \citet{nxt2016probabilistic} presents a probabilistic approach to reorg susceptibility in the Nxt protocol, and concludes that security in the model relies on the attacker having less than one-third of the the total staking power.  \citet{gasper} present a model of the Ethereum 2.0 beacon chain, a forthcoming PoS protocol. Their design combines the GHOST (``greedy heaviest-observed subtree'') fork choice rule \cite{sompolinsky2015secure} with the Casper consensus protocol \cite{buterin2017casper}. They prove liveness of the protocol both probabilistically and practically in the face of an adversary with less than one-third of the staking power (a constraint that is not imposed on the Tezos protocol), and minimize deep reorgs by enforcing finality through checkpoints.
Nomadic Labs, the research group that implemented the 2020 software update to Tezos, called {\em Carthage}, published a blog post analyzing the rate of forking in the Tezos consensus mechanism \cite{analysisemmyplus}. The blog post links to the code used to conduct their analyses and provides a high level intuition for the techniques used. Complementing this, we present an explicit formulation of the techniques and models used in our analysis and confirm that these two analyses match where 
they consider the same questions.

The present paper builds on \citet{neuder2019selfish}, who study selfish mining in 
Tezos and show that it is sometimes more profitable to create a length-2 reorg 
than to follow the honest protocol.
This previous work considered only selfish mining and not
the possibility of reorgs conducted for the purpose of 
double-spends. In the present paper, we focus instead on the rate at which malicious reorgs are feasible, without concern for protocol-reward based profit.
Coupled with double-spends, deep reorgs have 
the potential to generate a profit for an attacker that far exceeds protocol rewards.
We also develop a model of profitability for selfish mining of arbitrary length, instead of restricting analysis to length-2 forks,
by making use of Monte Carlo methods alongside importance sampling for sample efficiency.  

This work also relates to the literature on selfish mining.  First discussed in PoW chains by~\citet{eyal2018majority},  \citet{brown2019formal} study selfish mining for PoS, demonstrating that all longest-chain PoS systems are susceptible to {\em predictable selfish mining} and {\em predictable double-spend} attacks, by which they mean situations where attackers can determine ahead of time when either attack will be possible.
These attacks, hypothesized for abstract models of longest-chain PoS systems by \citet{brown2019formal}, are  modelled for the Tezos protocol in this work. 

\begin{figure*}
    \centering
    \includegraphics[width=0.55\textwidth]{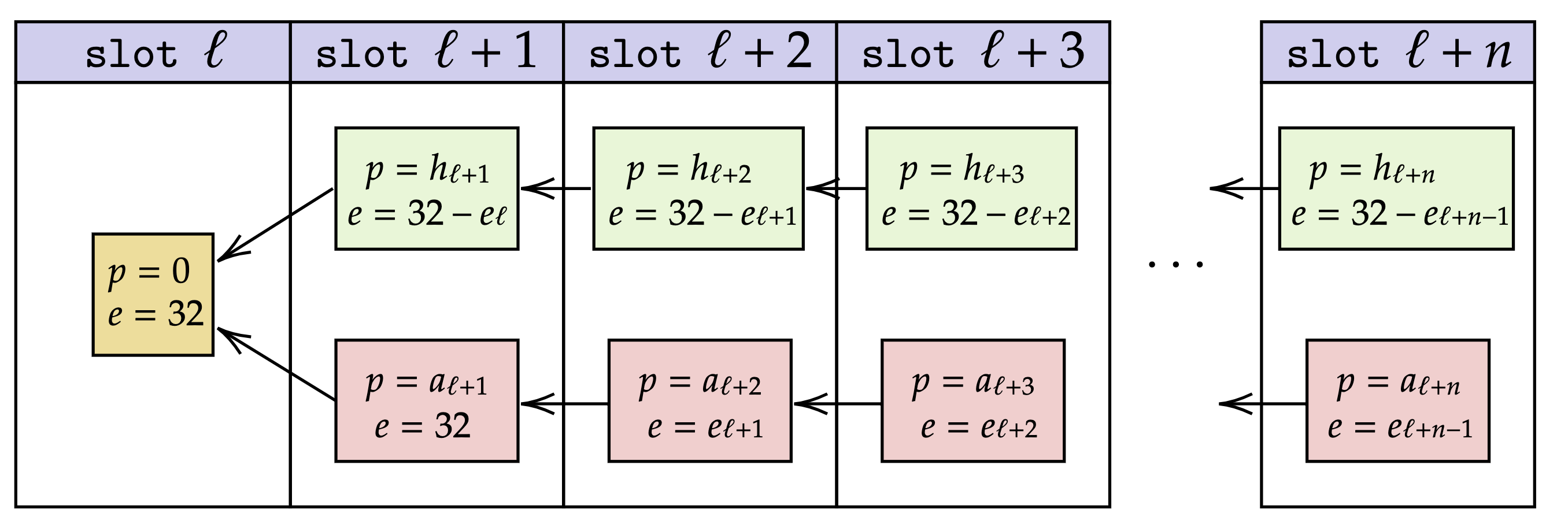}
    \caption{The progress of an honest (green blocks) and attacker (red blocks) chain for a length-$n$ malicious reorg.  The gold block on the left represents the last publicly agreed upon block, and the origin of the attacker's private fork. For each block
    we give the particular priority and number of endorsements with which it is baked. The state variables $\mathbf{a} = (a_{\ell+1}, a_{\ell+2}, ..., a_{\ell+n}), \mathbf{h} =(h_{\ell+1}, h_{\ell+2}, ..., h_{\ell+n}),\textrm{ and }\mathbf{e}=(e_{\ell}, e_{\ell+1}, ..., e_{\ell+n-1})$ denote the  sequence of attacker priorities, honest network priorities, and attacker endorsements, respectively. A length-$n$ malicious reorg is \textit{feasible}
    if the attacker can reach the $\ell+n$ block height before the honest network.
    \label{fig:tezosModel}}
\end{figure*}

\section{The Tezos Proof-of-Stake Protocol}
\label{sec:tezos-pos}

The Tezos blockchain launched in 2018 and had a market capitalization of \$1.8 billion as of September 2020, ranking 15th by this metric among cryptocurrencies. Tezos uses a longest-chain\footnote{Longest-chain refers to the fork choice rule that the honest nodes follow. When presented with conflicting forks, the honest nodes choose whichever has more blocks, as in the Bitcoin PoW implementation.} PoS protocol, and includes a number of distinguishing features such as on-chain governance and Turing-complete smart contract support \cite{coinmarketcap}. The model that we study relies on a synchronous network in which all honest operations are seen by all participants before the subsequent block is published. This is a common assumption  in the selfish mining literature \cite{eyal2018majority, sapirshtein2016optimal,neuder2019selfish, analysisemmyplus}. 

\subsection{The Protocol}

Participants in the Tezos PoS protocol put a certain amount of Tezos tokens into escrow smart contracts (staking) 
as collateral to incentivize their own honest behavior. In each block, a sample of these stakers is randomly selected to participate by filling one of two roles: {\em baker} or {\em endorser}. 
\begin{itemize}
\item Bakers collect transactions gossiped over the P2P network and assemble them into blocks to publish
(analogous to miner in Proof-of-Work protocols)
\item Endorsers  cryptographically sign the  best  block (specified below) they have heard about at a given height.
\end{itemize}

The protocol issues rewards to those selected.
At each block height, 32 stakers are chosen as endorsers. In order to ensure that a block is baked at each block height, a list of potential bakers is drawn, and the index of a staker in that list indicates the {\em priority} with which they can create a block. 

In order to allocate block creation and endorsement rights, Tezos implements a `follow-the-Satoshi' random token selection process. The currency is divided into {\em rolls} of 8,000 native tokens (XTZ). At each block height, rolls are drawn at random to determine who has the ability to create and endorse a block.  This is all the detail we need for our present model and analysis. For a more detailed description of the Tezos consensus layer, see \cite{postezos,neuder2019selfish}.

A block's validity depends not only on the validity of the transactions it contains but also  its timestamp. A block is  {\em valid} when a specified delay has elapsed between it and its predecessor. This delay is a function of the priority of the baker, $p$, (0 is highest priority, followed by 1, etc.) and the number of endorsements that the block includes, $e$. 
\begin{definition}[Validity]
In order for a block to be valid, its timestamp must be at least $\mathcal{D}(p,e)$ seconds greater than that of the previous block, where
\begin{align}\label{eq:delay}
    \mathcal{D}(p, e) = 60 + 40 \cdot p + 8\cdot\max(24-e,0).
\end{align}
\end{definition}
\smallskip

Both endorsers and bakers are rewarded for participation in the consensus layer. 
\begin{definition}[Reward]
The baker's reward, $\mathcal{R}_b(p,e)$, for a block with  $e$  endorsements  and  priority $p$, is
\begin{align}\label{eq:rewBlock}
    \mathcal{R}_b(p, e) &= 
    \begin{cases}
    1.25e & \textrm{if } p = 0 \\
    0.1875e & \textrm{otherwise}. \\
    \end{cases}
\end{align}

The endorser's reward,  $\mathcal{R}_e(p_i)$, given
the endorsement is included in a block baked by
a baker with priority $p_i$, is
\begin{align}\label{eq:rewEnd}
    \mathcal{R}_e(p_i) &= 
    \begin{cases}
        1.25 & \textrm{if } p_i = 0 \\ 
        0.8333333 & \textrm{otherwise}. \\ 
    \end{cases}
\end{align}
\end{definition}
\smallskip

Observe that the reward for baking a block without the highest priority is substantially smaller than a block with priority 0. 

We define a {\em length-$n$ malicious reorg} as the situation in which an attacker can create $n$ blocks faster than the rest of the network, with the effect that $n-1$ blocks are deleted from the public, or canonical, chain. 
This offers an opportunity for an {\em $n-1$ confirmation double-spend}, in which the attacker makes a transaction that is included on the public chain and then waits $n-1$ blocks before deleting the block that includes the transaction. 

\subsection{Malicious Reorgs}
It is helpful to define the {\em state} of the Tezos PoS protocol as it relates to the honest (public) chain and an attacker's chain.
\begin{definition}[Slot Configuration]
We use the following three variables to describe the  configuration of a slot at block height $\ell$:
\begin{itemize}
    \item $a_\ell$, the highest priority of the attacker for slot $\ell$.
    \item $h_\ell$, the highest priority of the honest network for slot $\ell$.
    \item $e_\ell$, the number of endorsement rights owned by the \textit{attacker} at slot $\ell$. 
\end{itemize}
\end{definition}
We can define the state of the Tezos PoS protocol by composing $n$ slots together.

\begin{definition}[State]
The state is made up of three sequences of length $n$.
\begin{itemize}
    \item  $\mathbf{a} = (a_{\ell+1}, a_{\ell+2}, ... , a_{\ell+n})$, the next $n$ attacker priorities
    \item $\mathbf{h} = (h_{\ell+1}, h_{\ell+2}, ... , h_{\ell+n})$, the next $n$ honest priorities
    \item  $\mathbf{e} = (e_{\ell}, e_{\ell+1}, ... , e_{\ell+n-1})$, the next $n$ attacker-owned endorsement counts
\end{itemize}
\end{definition}

Given this, let $\mathcal{S} = \{\mathbf{a}, \mathbf{h}, \mathbf{e}\}$ denote the
{\em state of the next $n$ blocks of the chain} after slot $\ell$. 
See Figure~\ref{fig:tezosModel}. 

Notice that~\eqref{eq:delay} is a function of the number of endorsements for the previous slot that the current block includes, so that the endorsements start at slot $\ell$ while the priorities start at slot  $\ell+1$. Additionally, for any given slot $i$, $a_i \textrm{ and } h_i$ are \textit{not} independent. This is because each priority must either be owned by the attacker or the honest network, but cannot be owned by both. 

\subsection{Feasibility Function}
We now develop an expression to determine if a given state $\mathcal{S}$ allows a \textit{feasible} attack. Let $\delta_a(\mathbf{a},\mathbf{e})$ denote the amount of time it takes an attacker to create $n$ blocks. Based on the delay function~\eqref{eq:delay}, we can calculate $\delta_a$ as
\begin{align}
    \delta_a(\mathbf{a}, \mathbf{e}) &= \mathcal{D}(a_{\ell+1}, 32) + \sum_{i=2}^n \mathcal{D}(a_{\ell+i}, e_{\ell+i-1}).
\end{align}

For the first block, the attacker is able to use all of the honest endorsements for the previous block. In contrast, since the attacker is selfishly endorsing their own private fork and not sharing those endorsements over the P2P network,
the honest network will only hear the $32-e_{\ell}$ endorsements that they were allocated at slot $\ell$. 
The expression for $\delta_h(\mathbf{h}, \mathbf{e})$, the time that it takes the rest of the network to create $n$ blocks, is,
\begin{align}
    \delta_h(\mathbf{h}, \mathbf{e}) &= \sum_{i=1}^n \mathcal{D}(h_{\ell+i}, 32-e_{\ell+i-1}). 
\end{align}

We do not need a state variable to keep track of the honest endorsement allocation because we know that at each slot $\ell$ the honest network is allocated $32-e_\ell$ endorsements. 

\begin{definition}[Feasibility of a length-$n$ malicious reorg]
The feasibility function for a length-$n$ malicious reorg is 
\begin{align}
    \F_n(\mathcal{S}) = \mathbb{I} [\delta_a(\mathbf{a}, \mathbf{e}) \leq \delta_h(\mathbf{h}, \mathbf{e})], 
\end{align}
where $\mathbb{I}[\texttt{condition}] = 1$ if and only if \texttt{condition} is \texttt{true}. 
\end{definition}

This function indicates if a reorg is feasible in state $\mathcal{S}$.

\subsection{Distribution on States in Tezos PoS}
In order to calculate the probability of an attack being feasible, we first determine the probability of an arbitrary state arising  on the chain. 

Let $\alpha$ denote the
{\em fraction of the total stake owned by the attacker}. 
The attacker's highest priority state variables, $a_i$, are distributed 
%\begin{align}
 $A\sim \textrm{Geometric}(\alpha)$,
where we define a geometric random variable as counting the number of failures until a success.\footnote{For example, if the highest priority an attacker owns for a particular block is 3, the probability of this occurring is, 
    $\mathrm{Pr}[A = 3] = (1-\alpha)^3 \alpha$,
because the priorities of $\{0,1,2\}$ must all be owned by the honest network. } 
The honest network's highest priority state variables,  $h_i$, are distributed
$H\sim \textrm{Geometric}(1-\alpha)$.
This is defined for the complement of $\alpha$, reflecting the amount of stake owned by the honest network.
The endorsement counts, $e_i$, are distributed 
$E \sim \textrm{Binomial}(32, \alpha)$,
reflecting the number of endorsement rights  owned by the attacker when the probability of owning any  one endorsement right is $\alpha$.

We consider the problem of calculating the probability of a given state, $\Scal = \{\mathbf{a}, \mathbf{h}, \mathbf{e}\}.$ 
By independence, the probability of $\Scal$ is  
$\mathrm{Pr}[\Scal] = \prod_{i=1}^n \mathrm{Pr}[\mathcal{S}_i]$,
where $\Scal_i = \{a_i, h_i, e_{i-1}\}$
denotes the state variables for slot $i$.
The last step is finding an expression for the value of $\mathrm{Pr}[\mathcal{S}_i]$.
Here, the events $a_i$ and $h_i$ are not independent
because the attacker  has the highest priority if and only if the honest network does not have priority of 0 for a slot. That is,
\if 0
\begin{align}\label{eq:iff}
    A = 0 \iff H \neq 0.    
\end{align}
Also notice that this is logically equivalent to,
\fi
$ H = 0 \iff A \neq 0$. 
For each slot, either $h_i$ or $a_i$ must be equal to zero, and the other must be non-zero. Using this, the probability of $\{a_i, h_i\}$ is,
\begin{align}
    \mathrm{Pr}[A=a_i,H=h_i] &= 
    \begin{cases}
    \alpha^{h_i} (1-\alpha) & \text{if } a_i=0 \\ 
    (1-\alpha)^{a_i} \alpha & \text{if } h_i=0 \\ 
    \end{cases}.
\end{align}

Since the  number of endorsement slots owned by the attacker is independent of the priority list,
the probability of $\mathcal{S}_i$ is
$\Pro[\mathcal{S}_i] = \Pro[A=a_i,H=h_i] \cdot \Pro[E=e_i]$,
and we have 
\begin{align}
    \Pro[\mathcal{S}] = \prod_{i=1}^n \Pro[A=a_i,H=h_i] \cdot \Pro[E=e_{i-1}].
    \label{eq:probstate}
\end{align}

We use~\eqref{eq:probstate} to validate the results of  Monte Carlo estimation as well as to  calculate likelihood ratios for use in importance sampling, as in \ref{subsec:attack-probs} and Appendix \ref{app:analyticTable}.

\begin{figure*}
    \centering
    \includegraphics[width=0.7\textwidth]{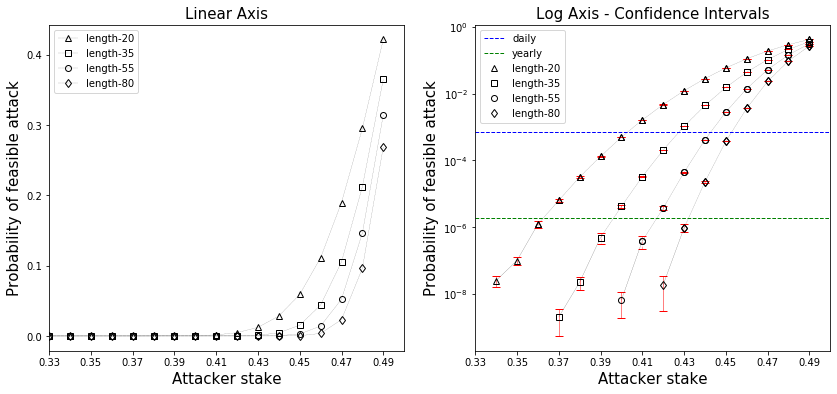}
    \caption{The probability of a  feasible reorg on Tezos PoS for different attacker stakes and lengths of reorg, linear axis (left) and log axis (right). 
    The right-hand plot also shows the probability that corresponds to a daily attack
    and an annual attack, in blue and green, respectively. The right figure also displays the 99\% confidence interval for each estimator as red errorbars. For probabilities above $10^{-7}$ these are Clopper-Pearson intervals, and for lower probabilities we use importance sampling. 
    \label{fig:highalphas}}
\end{figure*}

\subsection{Calculating Cost}\label{sec:cost}

Another factor of interest in the analysis of Tezos is 
the cost of an attack. Because the attacker is playing a dishonest policy, the rewards earned over the next $n$ blocks will differ from what they would have been had the attacker played honestly.

Given a feasible state $\mathcal{S}$, we provide an expression to calculate the cost of executing a length-$n$ reorg. Let $\phi_h(\mathcal{S})$ be the reward for the attacker playing honestly over the next $n$ blocks defined by state $\mathcal{S}$, and $\phi_a(\mathcal{S})$ be the reward for the attacker executing the reorg policy. These values are based only on the block reward, and do not correspond to the potential profit from performing a double-spend. Then we define the cost of an attack as  $\mathcal{C}(\mathcal{S}) = \phi_h(\mathcal{S}) - \phi_a(\mathcal{S})$.

There is no guarantee that the cost is positive. Rather, states in which the attacker earns more by executing the reorg are susceptible to selfish mining~\cite{neuder2019selfish}. We calculate $\phi_a$ and $\phi_h$ in Tezos as the sum of the block and endorsement rewards earned from the respective policies over the course of the next $n$ blocks (Equations \ref{eq:rewBlock} \& \ref{eq:rewEnd}). 
The reward of the $i^{th}$ slot of the state, $\mathcal{S}_i$, under the honest policy  is, 
\begin{align}
    \phi_{h,i}(\mathcal{S}_i) &= 
    \begin{cases}
        \mathcal{R}_b (0, 32) + e_{i-1} \mathcal{R}_e(0) & \textrm{if } a_i = 0 \\ 
        e_{i-1} \mathcal{R}_e(0) & \textrm{if } a_i \neq 0 .
    \end{cases}
\end{align}

In this case, the attacker will only earn a block reward if they have priority $0$ for that block. Otherwise, they simply earn the endorsement value for ending up on the highest priority block. Given this, the total reward for following the honest policy for the next $n$ blocks is,
\begin{align}
    \phi_h(\mathcal{S}) = \sum_{i=1}^{n} \phi_{h,\ell+i}(\mathcal{S}_{\ell+i}). 
\end{align}

The reward for following the attacking policy for the $i^{th}$ slot is,
\begin{align}
    \phi_{a, i}(\mathcal{S}_i) &= \mathcal{R}_b(a_i, e_{i-1}) + e_{i-1}\mathcal{R}_e(a_i).
\end{align}

When the $n-1$ blocks are deleted, each of the $n$ blocks on the attacker fork are accepted, and thus earn a block reward. In addition, the attacker's endorsement reward is  parameterized with the value of $a_i$. The total reward earned by following the attacking policy for $n$ blocks given state $\mathcal{S}$ is,
\begin{align}
    \phi_a(\mathcal{S}) &= \mathcal{R}_b(a_\ell, 32) + e_{\ell-1}\mathcal{R}_e(a_\ell) + \sum_{i=2}^n \phi_{a,\ell+i}(\mathcal{S}_{\ell+i}).
\end{align}

As with the feasibility function, the index  starts at $2$ because the initial block of the attacker fork will include all 32 endorsements. Combining, the cost for an attack parameterized by the state $\mathcal{S}$ is
$\mathcal{C} (\mathcal{S}) = \phi_h(\mathcal{S}) - \phi_a(\mathcal{S})$.
This provides the cost of an attack given that it is feasible to execute a length-$n$ malicious reorg. We use this expression to calculate the cost of deep reorgs, as well as identify states that are susceptible to selfish mining.

\subsection{Estimating Attack Probabilities}\label{subsec:attack-probs}

We use three techniques to estimate the probability of a reorg being feasible. For small reorgs,
we can estimate probabilities directly, by enumerating a large subset of the high probability states and counting those which allow feasible reorgs. Because the number of possible states (permutations) grows exponentially in fork length, this enumeration technique quickly becomes computationally intractable
(see Appendix \ref{app:analytic} and~\ref{app:analyticTable}). 

Mainly, we use Monte Carlo estimation, along with importance sampling for variance reduction.
Monte Carlo estimation provides  an unbiased estimator for the probability of a length-$n$ reorg by randomly generating states according to the attacker's stake, $\alpha$, and dividing the number of feasible states by the sample size. This works  well in cases where the probability of an attack is relatively large.

When trying to estimate the probability of a rare event, however, extremely large sample sizes are necessary in order to obtain a tight confidence interval. 
For better sample efficiency, we use importance sampling, which relies on defining a new distribution in which the event of interest  occurs more frequently, and then weighting the value in the sum by the likelihood ratio, which is the probability that the event occurs under the original distribution divided by the probability that it occurs under the modified distribution. See Appendix \ref{app:mc} and Appendix \ref{app:impsample} for details, along with a comparison with direct calculations (Appendix \ref{app:analyticTable}) and a demonstration of the variance reduction obtained through importance sampling (Figure \ref{fig:importancesamp}, Appendix \ref{app:impsample}).

\section{Malicious Reorgs in Tezos PoS}
\label{sec:results}

We now present the probabilities with which deep reorgs are feasible
in Tezos. 
As outlined in Section \ref{subsec:attack-probs}, we approximate these probabilities using a standard Monte Carlo method as well as importance sampling for events occurring with probability less than $10^{-7}$. 

Figure \ref{fig:highalphas} shows the  probability of a feasible reorg on Tezos PoS, along with the confidence intervals of the attack probabilities for length 20, 35, 55, \& 80 reorgs. In Figure \ref{fig:highalphas} (left), we see that the probability of a feasible attack is relatively low until around $\alpha=0.42$, above which the probability of feasible attack increases rapidly.

Figure \ref{fig:highalphas} (right)  uses a log y-axis,
and the blue and green lines correspond to the probability level at which 
we would see an attack in expectation, once daily and once yearly respectively. 
Notice that with an attack stake of about 36\%, an attacker can expect to do a length-20 reorg once a year, and with  an attack stake of about 40\%, an attacker can expect to do a length-20 reorg once a day. 

Beyond feasibility, deep reorgs are cheap to execute in terms of amount of block reward lost relative to playing an honest strategy. By calculating the average cost of a feasible attack through Monte Carlo simulations, we find that all attacks less than or equal to length 32 cost under 305 XTZ ($\approx \$700$ as of April 2020), with longer attacks being similarly cheap. 

These costs are small relative to the large potential profit from a successful double-spend. The primary challenge in implementing a deep reorg is in obtaining enough stake to create a feasible attack.

\section{Protocol Modifications}\label{sec:protocol-mod}

\begin{figure*}
    \centering
    \includegraphics[width=0.7\textwidth]{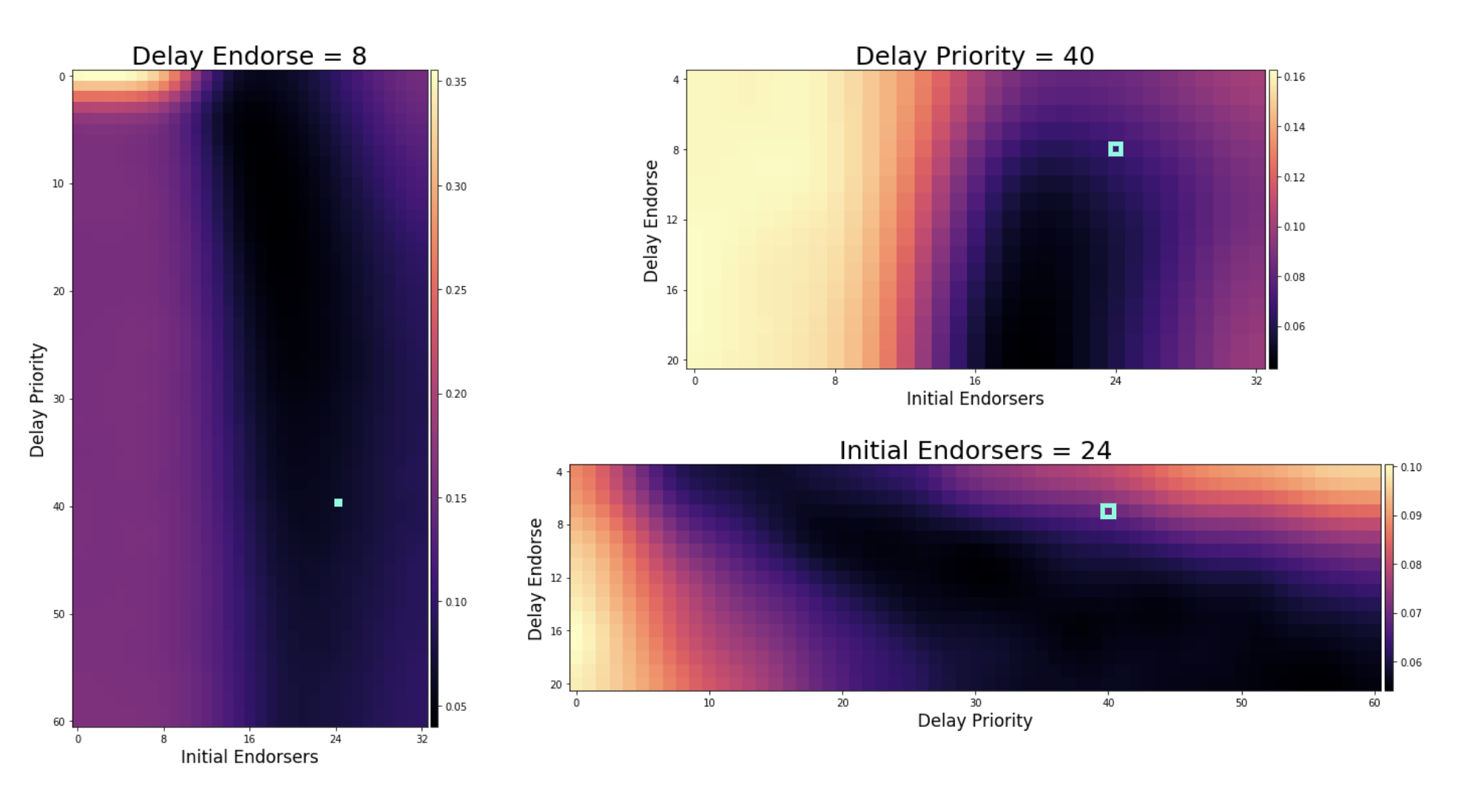}
    \caption{Slices of the 
    design space, optimized here for $\beta=0.5$, which considers an objective that seeks to minimize the average probability of a  length-20 reorg and a length-3 selfish-mine for an attacker with stake $\alpha=0.45$. Lower (darker) values 
    are desirable. For each slice, we fix the third design parameter to that in the Tezos PoS protocol, and the blue boxes correspond to the Tezos design $\xi = (24, 8, 40)$.
    \label{fig:diffparams}}
\end{figure*}

In this section, we consider what would happen if the protocol parameters in
the Tezos PoS implementation were modified.
This is in the same vein as the ``Choosing the Constants'' portion of the Emmy+ analysis \cite{analysisemmyplus}, which considers the probability of selfish mining and deep reorgs separately and heuristically chooses parameters based on a visual analysis of the figures.

We adopt the same protocol parameters, but  expand on their analysis, giving a well-defined objective and exploring a larger domain of protocol parameter combinations.
To choose better protocol parameters, we consider minimizing a weighted sum of the probability of selfish mining and deep reorgs concurrently.

Although both deep reorgs and selfish mining result in deleting blocks from the public chain,
these two objectives  have opposite effects on the choice of protocol parameters. 
This surprising phenomenon comes about  because selfish behavior occurs on a much shorter scale than the long reorgs that offer double-spend opportunities. 

To understand this, note that an endorsement included on a low priority block is worth much less than one on the highest-priority block~\eqref{eq:rewEnd}. For this reason, each time an attacker builds a low priority block using their private set of endorsements, each of the endorsements loses considerable value. This causes the profitable forks to occur overwhelmingly at shorter block lengths (e.g., an attacker with $\alpha=0.45$ never finds a profitable selfish-mine of length 20, while if they are looking just to do a reorg of that length, they have a probability of $0.1$ per block). 
As a result, an attacker looking to execute a selfish-mine prefers an extremely high value of \texttt{Initial Endorsers} (as described below), because if they have a few blocks with many of the top priorities, they can slow the honest network down by withholding relatively few endorsements. 
Now consider instead an attacker trying to perform a long-range reorg. With a high value of \texttt{Initial Endorsers}, the attacker will have to pay a time penalty for each missing endorsement, and thus each block will be slower, making long range forks nearly impossible. This demonstrates why optimizing the protocol to resist only deep reorgs or only selfish mining is unsuccessful. A safe protocol must find a balance between these two attack vectors. 

\subsection{Protocol Parameters}

In our analysis, we  keep the structure of the Tezos PoS protocol unchanged, but alter the following three protocol parameters, which are the same as those modified by Nomadic Labs \cite{analysisemmyplus}.
\begin{enumerate}
    \item \texttt{Initial Endorsers (default=24):} This is the minimum number of endorsements needed to not be penalized for missing endorsements. Denoted $e_i$ ($\mathit{ie}$ in \cite{analysisemmyplus}).
    \item \texttt{Delay Endorse (default=8):} The time penalty, in seconds, accrued for each missing endorsement below \texttt{Initial Endorsers}. Denoted $d_e$ ($\mathit{de}$ in \cite{analysisemmyplus}).
    \item \texttt{Delay Priority (default=40):} The time penalty, in seconds, accrued for each drop in priority below 0. Denoted $d_p$ ($\mathit{dp}$ in \cite{analysisemmyplus}). 
\end{enumerate}
\begin{figure*}[t!]
    \centering
    \includegraphics[width=0.7\textwidth]{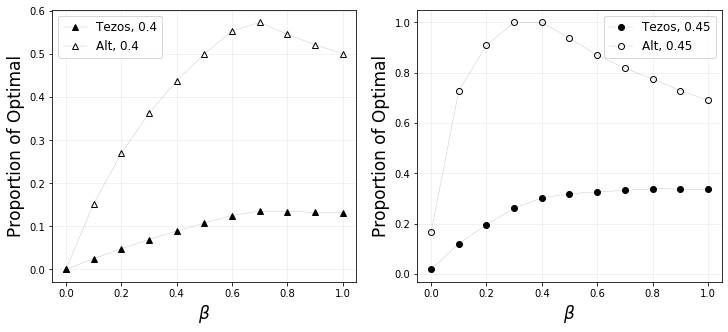}
    \caption{ The relative performance of the Tezos PoS protocol, $\xi = (24, 8, 40)$, 
    compared to a design with parameters $\xi=(15,5,8)$, 
    for both $\alpha=0.4$ and $\alpha=0.45$.
For example, a value of $0.4$ means that the performance of the Tezos PoS design is $2.5$ ($2.5=1/0.4$) times larger than the minimum (optimal for $\alpha$) value. 
For both choices of $\alpha$, the alternative design outperforms the Tezos PoS design 
for all values of $\beta$. The left figure was calculated with $n_1=3 \;\&\; n_2=8$ and the right figure with $n_1=3 \;\&\; n_2=20$.
\label{fig:weightingTraces}}
\end{figure*}

The default values are from the Carthage version of the Tezos PoS specification \cite{cathagedoc}.
With these parameters, we use the same general formulation of delay as Nomadic Labs \cite{analysisemmyplus}: 
\begin{align}
    \mathcal{D}(p, e) &= 60 + d_p \cdot p + d_e \cdot \max (e_i - e, 0).
\end{align}

By varying the values of these parameters, we  examine how the frequency of deep reorgs and selfish-mines changes as the protocol parameters are altered.

\subsection{Balancing Malicious Reorgs and Selfish Mining} 

We optimize the design for a given value of $\alpha$ (say $\alpha=0.45$ or $\alpha=0.4$), because large stake attackers have the highest probability of successfully launching reorg attacks. 
We also fix two lengths, $n_1$ and  $n_2$, for the deep reorg and selfish-mine respectively.

In order to ensure the objective function weights the probabilities relatively evenly, we choose a deep reorg length of $20$ and a selfish-mine of length 3 which, for $\alpha=0.45$, have approximately the same probability of occurring. This ensures that the design objective 
considers both kinds of attacks. 
We  vary  the choice of  objective weights between the 
probability of malicious reorgs and the probability of selfish mining.

Let $\mathcal{F}_n(\mathcal{S})$ denote the feasibility function of a reorg of length-$n$.
%(the main focus of this work).
Let $\mathcal{P}_n(\mathcal{S})$ denote a function describing if a reorg of length-$n$ is \textit{profitable} and \textit{feasible} in the given state (e.g., a selfish-mine state). 
We represent the Tezos PoS design parameters by the tuple $\xi = (e_i, d_e, d_p)$,
and for a given attacker stake $\alpha$, consider the objective 
\begin{align}
    \mathcal{O}(\xi,\beta) &= (1-\beta)\underbrace{\mathrm{Pr}[\mathcal{F}_{n_1}(\mathcal{S}) | \alpha, \xi]}_{\mathcal{O}_1} + \beta \underbrace{\mathrm{Pr}[\mathcal{P}_{n_2}(\mathcal{S})  |  \alpha, \xi]}_{\mathcal{O}_2},
%    \mathcal{O}(\xi, \beta) &= (1-\beta) \cdot \mathcal{O}_1 + \beta \cdot \mathcal{O}_2.
\end{align}
where parameter $\beta\in[0,1]$  controls the relative weight given to 
the considerations of $\mathcal{O}_1$ and $\mathcal{O}_2$,
corresponding to the probability of a deep reorg of length $n_1$ and a selfish-mine of length $n_2$, respectively.

For the choices of $n_1=20$ and $n_2=3$, and with $\alpha=0.45$,  the probabilities of reorgs and selfish mining are $0.05974$ and $0.07243$, respectively. 
We also demonstrate similar results with $\alpha =0.4, n_1=8$, and $n_2=3$ in Figure \ref{fig:weightingTraces}. 

We consider the following choice of PoS design parameters:
\begin{itemize}
    \item \texttt{Initial Endorsers}: $e_i\in [0, 32]$
    \item \texttt{Delay Endorse}:  $d_e\in [4, 20]$
    \item \texttt{Delay Priority}: $d_p\in [0, 60]$
\end{itemize}

These are chosen to provide a large range around the current Tezos implementation of $\xi = (24, 8, 40).$

\subsection{PoS Protocol Design Results}
Figure \ref{fig:diffparams} demonstrates three slices of the resulting 3D array of the objective function, for the choice of $\beta=0.5$.
  Figures~\ref{fig:diffparams} and~\ref{fig:weightingShift} 
 include a Gaussian filter  to reduce noise and display the underlying structure more clearly.
 Each slice fixes the value of the current Tezos implementation for 
 one of the variables, and the blue square in each highlights the current Tezos combination of $\xi = (24, 8, 40)$. 

Figure \ref{fig:diffparams} 
illustrates the structure of $\mathcal{O}$, for $\beta =0.5$,
as a function of the design parameters. Low values (dark) are preferable.
The Tezos implementation has a relatively low value.
However, there are better choices for design parameters $d_p$ and $d_e$, and these also have the added benefit of keeping the block creation rate fast, and in some cases, even speeding it up.

Figure \ref{fig:weightingTraces} compares the current Tezos parameterization with an alternative PoS protocol that makes use of design parameters $\xi=(15,5,8)$,
varying $\beta$ on the x-axis and showing results 
for both attacker stake $\alpha=0.4$ and $\alpha=0.45$. 
At each value of $\beta$, we plot the ratio of the minimum value of the objective function to the value of either the alternative or current protocol parameters. For example, a value of $0.4$ means that Tezos PoS achieves an objective that is 2.5 times larger than the minimum value acheived by any $\xi$. Sometimes, the reduction in probability is far greater. For example, an $\alpha=0.4$ attacker can expect about 24 length-10 reorgs a day using the current Tezos implementation, but only 1 per-day under our modified protocol. Further, that same attacker could expect about one length-20 reorg per day under the Tezos protocol, but only one per year under the modified system. 

Figure~\ref{fig:weightingShift} illustrates
the way in which different weightings  change the value of the objective function. 
Each cell represents a specific value of design parameters $\xi$. 
In this case, we give results for a choice of $\beta=0$, $0.7$, and $1$. 
Figure~\ref{fig:weightingShift} (top) only 
cares about malicious reorgs,
the objective in Figure~\ref{fig:weightingShift} (middle) cares about both,
and Figure~\ref{fig:weightingShift} (bottom) only cares
about selfish mining. 
As the weighting changes, the landscape essentially flips from high values on the left to high values on the right. This demonstrates that minimizing for the two objectives can pull the protocol parameters in opposite directions, confirming that need both
to be considered for choosing the rules of PoS. 
\begin{figure}
    \centering
    \includegraphics[width=0.4\textwidth]{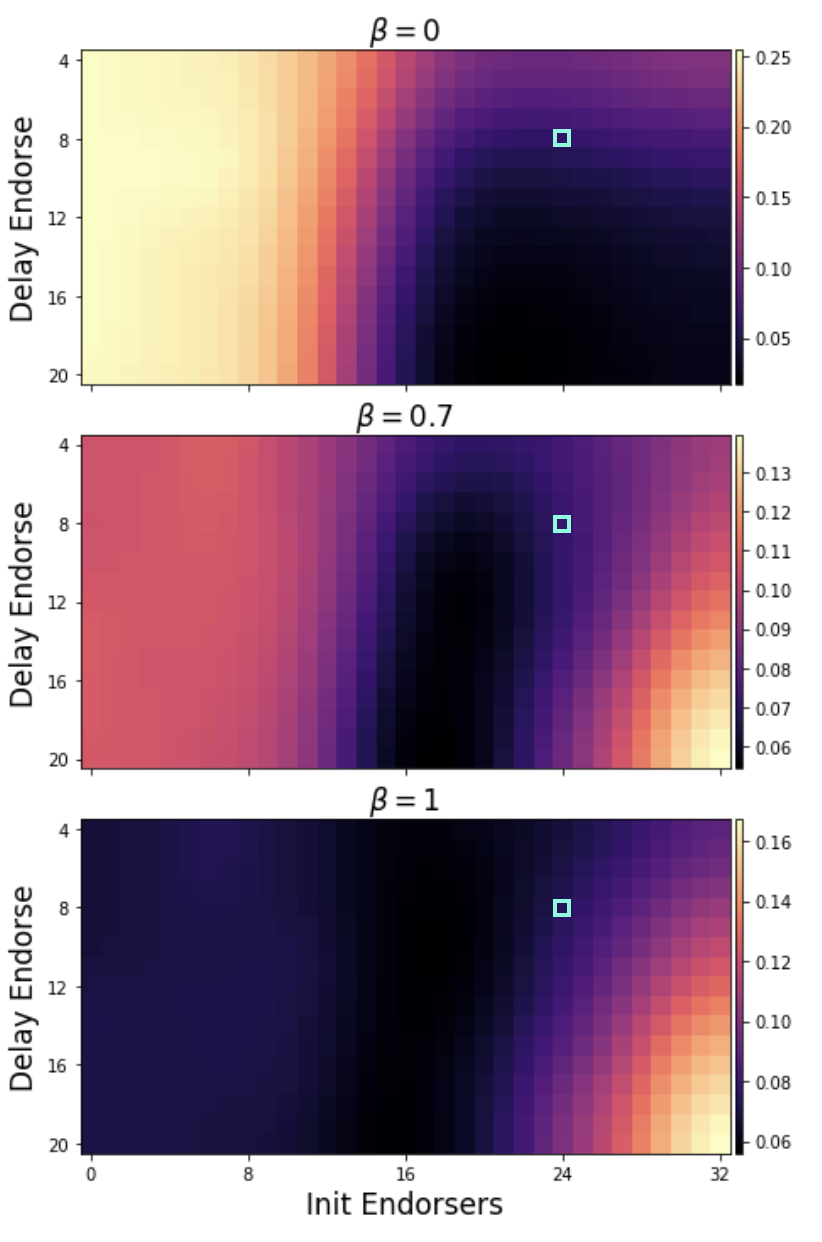}
    \caption{ The relative performance of the Tezos PoS protocol for choices
    of $\beta=0$, $\beta=0.7$, and $\beta=1$ (top, middle, bottom).  Each 
    problem of optimal design 
    fixes \texttt{Delay priority = 40}, which is the choice
    made in the  Tezos PoS protocol and the blue boxes correspond to the Tezos design.
    The shape of the landscape completely flips as a result of different considerations
    of objective weight $\beta$. 
    \label{fig:weightingShift}}
\end{figure}

The current Tezos implementation appears to perform far worse than the alternative. However, we caveat this result with two important features of the objective function. First, the regions of low values on the objective landscape are large. Figure \ref{fig:diffparams} presents this visually with relatively widespread dark regions that imply that many possible parameter configurations will do well to keep the objective low. Second, this only considers a single value of $\alpha$.

To understand the effect of different stake amounts, Figure \ref{fig:weightingTraces}  shows an alternative value of $\xi$ outperforming Tezos both at $\alpha=0.4$ and $\alpha=0.45$, but these are still high values of $\alpha$. It is plausible that the current Tezos implementation outperforms the alternative when the attacker stake and thus the probabilities of attacks are far lower.

\section{Detecting Malicious Reorgs in  Tezos PoS}
\label{sec:detection}

\begin{figure*}
    \centering
    \includegraphics[width=0.7\textwidth]{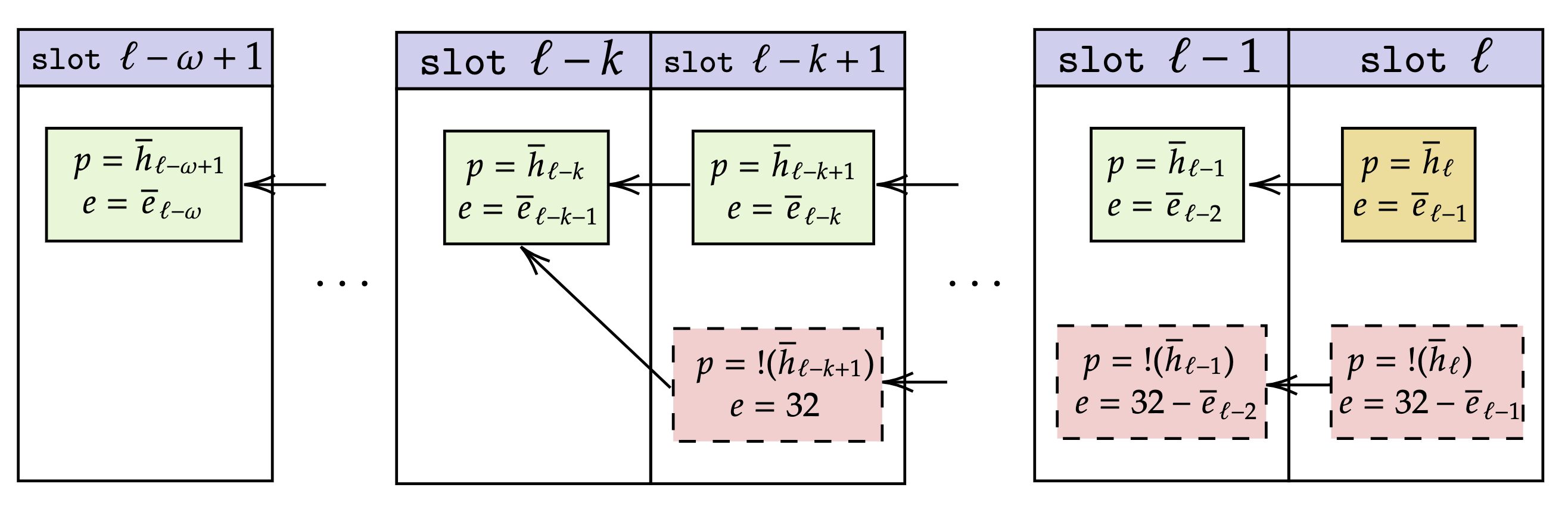}
    \caption{This figure helps explain the intuition behind our health metric for Tezos PoS. The gold block represents the current public head of the chain, and the green blocks show the last $\omega$ public blocks. The red-dashed blocks represent potential blocks that could be created by an attacker with the highest available priority and all the missing endorsements for the previous $k$ blocks. The priorities of the potential blocks are $!(\bar{h}_j)$, where $!$ maps all non-zero elements to 0 and all 0 elements to 1 (e.g., $!(5)=0$, while $!(0)=1$). The highest available priority at slot $\ell-j$ is  0 if the block on the honest chain was not baked with 0 priority, and 1 otherwise. $\bar{\delta}_h$ is the amount of time take for the honest network to create blocks in slots $\ell-k+1$ through $\ell$, and $\bar{\delta}_a$ is the same calculation  for the potential blocks.
    \label{fig:detection}}
\end{figure*}
\begin{figure*}[t!]
    \centering
    \includegraphics[width=0.7\textwidth]{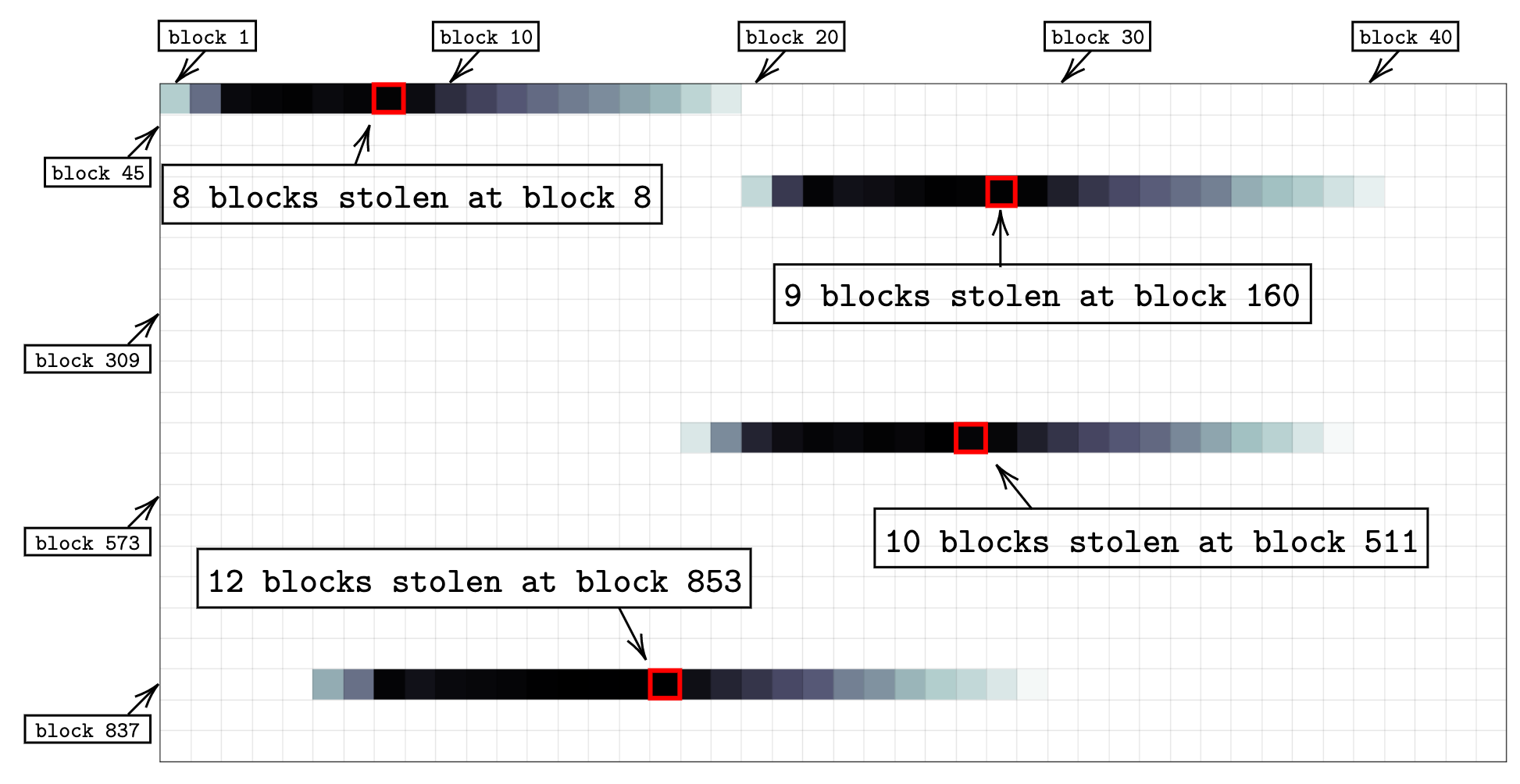}
    \caption{The value of the health metric, $\mathcal{H}(\bar{S})$, over the course of 968 simulated blocks. Each square of the grid represents a block, and the order in time is from left to right, and from top to bottom. The block health is represented by the color, with white being perfectly healthy and black being perfectly unhealthy. The simulation 
    considers an attacker with  $\alpha=0.375$ stake  and who  considers attacks of length 8 or greater. The red squares indicate blocks at which the attacker carries out a reorg, and 
    are labelled with the block at which the reorg occurred and the number of blocks that are deleted. In each case, the health metric immediately detects the malicious behavior, and signifies an unhealthy chain all the way through the attack.
    \label{fig:simulation}}
\end{figure*}

In this section, we introduce a method to detect malicious reorgs that are in progress on the Tezos blockchain. The calculations that we provide correspond to the current Tezos PoS protocol, but the method can also be used for other design parameters.
This metric allows vulnerable states of the Tezos chain to be identified, and users can know when to be careful with large transactions. It doesn't eliminate the concern about deep reorgs, mainly because it is difficult to get large adoption of the metric within the community. Further, if no transactions are being accepted during a deep reorg, an attacker can effectively execute a denial of service attack, which is also highly undesirable.

We denote the state variables with overset bars to indicate that the variable refers to slots in the \textit{past} rather than the future (e.g., $\bar{h}_{\ell-k}$ denotes the priority with which a block at slot $\ell-k$ was baked by the honest network). 
Let $\omega$ define the {\em window of security}, or the number of blocks in the past that we will consider. Intuitively, this corresponds to the maximum length attack that is accounted for, and we only consider forks starting at slots $(\ell-\omega+1, \;\ell-\omega, \;..., \;\ell)$ given the current public head is at slot $\ell$. Denote the state of the past $\omega$ blocks as  $\bar{\mathcal{S}}$. This is different from the state as viewed by the attacker, 
because $\bar{\mathcal{S}}$ is backward-looking in time rather than forward-looking, 
and only a function of things that have already been committed to the chain and thus are available publicly. 
Let $\bar{\mathbf{h}} = (\bar{h}_{\ell-\omega+1}, \bar{h}_{\ell-\omega}, ..., \bar{h}_{\ell})$ denote the priorities with which the last $\omega$ blocks were published, and $\bar{\mathbf{e}} = (\bar{e}_{\ell-\omega}, \bar{e}_{\ell-\omega+1}, ..., \bar{e}_{\ell-1})$ denote the number of endorsements over the same sequence of blocks. 
The current state at slot $\ell$ is denoted $\bar{\mathcal{S}} = \{\bar{\mathbf{h}},\bar{\mathbf{e}}\}$. 

Figure~\ref{fig:detection} shows the {\em forking possibilities} (e.g., blocks at which the attackers fork could have started) that the detection metric uses to determine the health of the state, $\bar{\mathcal{S}}$. 
Let $\bar{\mathcal{S}}_j$
 denote the $j^{th}$ slot in the past.
 Let  $\bar{\delta}_{h}(\bar{\mathcal{S}}, k)$ denote the minimum time  needed to create the past $k$ blocks using the priority levels and endorsement counts found on the past $k$ blocks on the chain (e.g., the last $k$ values of the sequences $\bar{\mathbf{h}}$ and $\bar{\mathbf{e}}$), i.e.,
\begin{align}
    \bar{\delta}_{h}(\bar{\mathcal{S}}, k) &= \sum_{j=0}^{k-1} \mathcal{D}(\bar{h}_{\ell-j}, \bar{e}_{\ell-j-1}).
\end{align}

Let $\bar{\delta}_a(\bar{\mathcal{S}}, k)$ denote the minimum time it would take an attacker to create the last $k$ blocks while using the best remaining priority and all the missing endorsements from the public chain (e.g., if the public block at slot $\ell-k$ was baked with priority $0$ and $22$ endorsements, then the best block the attacker could create would be with priority $1$ and $10$ endorsements). 

The time it will take the attacker to create a block at slot $j$ in the past is, 
\begin{align}
    \bar{\delta}_{a,j}(\bar{\mathcal{S}_j}) &= 
    \begin{cases}
        \mathcal{D}(0, 32-\bar{e}_{\ell-j-1}) & \textrm{if } \bar{h}_{\ell-j} \neq 0 \\
        \mathcal{D}(1, 32-\bar{e}_{\ell-j-1}) & \textrm{otherwise.} \\
    \end{cases}
\end{align}

The total time it would take the attacker to create the last $k$ blocks is,
\begin{align}
    \bar{\delta}_{a}(\bar{\mathcal{S}}, k) &=
    \begin{cases}
        \mathcal{D}(0, 32) + \sum_{j=1}^{k-1}  \bar{\delta}_{a,j}(\bar{\mathcal{S}_j}) & \textrm{if } \bar{h}_{\ell-k+1} \neq 0 \\ 
        \mathcal{D}(1, 32) + \sum_{j=1}^{k-1}  \bar{\delta}_{a,j}(\bar{\mathcal{S}_j}) & \textrm{otherwise.}
    \end{cases}
\end{align}

The first block on the attacker's chain will include all 32 endorsements because the endorsements for the last common block will be valid. The rest of the blocks on the attacker chain will only include $32 - \bar{e}_i$ endorsements for each slot $i$. 

We are interested in the difference between these two times, $\Delta_k(\bar{\mathcal{S}}) = \bar{\delta}_{a}(\bar{\mathcal{S}}, k) - \bar{\delta}_{h}(\bar{\mathcal{S}}, k)$.
%\end{align}
%
If $\Delta_k(\bar{\mathcal{S}})$ is small,  an attacker could have created $k$ valid blocks  close to the time it took the honest network, and thus could overtake the public chain and perform a reorg in the future. Now we consider the value of $\Delta_k(\bar{\mathcal{S}})$ that corresponds to a perfectly healthy chain.
Intuitively, under ideal chain conditions (e.g., highest priority of each block and all 32 endorsements included), $\Delta_k(\bar{\mathcal{S}})$ grows linearly with regard to $k$, because each block created is a constant time faster than the next fastest block, and $\Delta_k(\bar{\mathcal{S}}) \propto k$.
Dividing $\Delta_k(\bar{\mathcal{S}})$ by $k$ gives a value that corresponds to the  health of the chain. This leads us to a health metric for the Tezos blockchain which describes how probable a reorg is in the future.
\begin{definition}[Health metric]
Given the priorities, $\bar{\mathbf{h}}$, and endorsement counts, $\bar{\mathbf{e}}$, of the public chain over the last $\omega$ blocks, we define the health of the chain, $\mathcal{H}$, as 
\begin{align}
    \mathcal{H}(\bar{S}) &= 
    \begin{cases}
        \min_{k < \omega}\frac{\Delta_k}{k} & \textrm{if } \min_k \Delta_k> 0 \\ 
        0 & \textrm{otherwise}.
    \end{cases}
\end{align}
\end{definition}

This  metric $\mathcal{H}(\bar{S})$ 
returns zero if the minimum value of $\Delta_k$ is less than zero. 
In this case, the attacker can already override the honest chain, because 
the attacker 
can create $k$ blocks faster than the honest network.  
The minimum value of the metric is zero,  corresponding to the lowest possible chain health.
We take the minimum value of the ratio with respect to $k$ because this represents the block in the past from which there is the highest risk of a fork. 

The relationship with the delay function (\ref{eq:delay}) makes the health metric  highly sensitive to missing block priorities, which each incur a 40 second time penalty. On the other hand, a few missing endorsements has a  minimal effect on the metric because this does not slow down the honest network. 

Through an analysis of 10,000 blocks, we find that 9,932 of them were baked with priority 0, and thus several low priority blocks being baked in a row is a strong signal of an attack and  will  be immediately detected. Additionally,  the endorsement counts are distributed normally with a mean of 22.26 and a standard deviation of 2.28. Under these  conditions, the metric will report a low risk of forks, because no potentially faster forks can be made with such limited withholding of endorsement and priorities. 

The health metric assumes that the attacker is not double baking or endorsing, which refers to a single delegate creating or endorsing multiple blocks of the same height. Double baking would only be of use to the attacker if they were able to split the honest majority, but based on our model, where both honest and attacker operations are instantly heard over the network, this is impossible. Thus the only way in which the attacker can slow down the honest chain is to withhold endorsements and blocks with high priority. Thus, participating in any way on the honest fork would only decrease the time between blocks for the honest fork and would not speed up the selfish fork, and thus would be of no use to the attacker.

\subsection{Simulations}
In order to demonstrate the effectiveness of the health metric, $\mathcal{H}(\bar{S})$, 
we simulate the network with an attacker executing malicious reorgs. 
The metric can also be used to detect selfish mining attacks by looking for endorsement and priority withholding.

Figure~\ref{fig:simulation} shows the health metric over the course of 968 simulated blocks. Each cell represents a block on the honest chain, and they are ordered from left to right and top to bottom (e.g., top left cell is block 1, and bottom right cell is block 968). Low values of health are indicated by the darker values of cells. In the simulation there is  an attacker controlling $\alpha=0.375$, and we assume the attacker only considers reorgs of length 8 or greater, and gives preference to longer range attacks (e.g., if both length 8 and length 10 reorgs are possible, the attack of length 10 is chosen).

The red boxes indicate the blocks at which the attacker carries out a reorg, and are labelled with the index of the block where the deletion occurred, as well as the number of honest blocks that were removed. 
From these results, we see that the health metric immediately detects the malicious behavior and indicates an extremely unhealthy chain all the way through the attack. Notice that as the attacker begins to withhold blocks on a private chain, the colors of the cells quickly transition from white to black, indicating that an attack is becoming very likely.

\section{Conclusion}
\label{sec:conclusion}

We have presented an analysis of the possibility of malicious reorgs on the Tezos protocol in which we have formulated statistical techniques for determining the degree of the protocol's vulnerability to malicious reorgs and demonstrated that a different choice of parameters can reduce the probability of attack.
In considering alternate configurations of the Tezos PoS protocol, we evaluated an objective function with different weights assigned to the probability of selfish mining versus that of malicious reorgs. 
We also provided a method for detecting malicious reorgs by defining a health metric that makes use of publicly available information to identify vulnerable chain states. This metric can be employed by users of the Tezos blockchain, including merchants, to determine how long to wait before confirming large transactions as valid. 

There are many interesting future directions to this work. The simplest may be to apply similar techniques to other protocols (perhaps attempting to extend analysis to protocols that rely on fork choice rules that do not use the longest-chain heuristic), to evaluate their security and find network parameters that optimize for transaction safety. One may also consider ensemble attacks that combine malicious reorgs with attacks on the chain's peer-to-peer communication network, 
which may prove to be more effective than reorgs in isolation. Further, we only consider alternative delay functions with the same functional form as the current Tezos implementation. It would be interesting to examine how other kinds of delay and reward functions, potentially of non-linear form, impact the probability of reorgs.

\begin{acks}
The authors would like to thank Lacramioara Astefanoaei, Eugen Zalinescu, Arthur Breitman, and Bruno Blanchet, for helpful discussions and comments. 
Additional thanks are due to the three anonymous reviewers. This work is supported in part by two generous gifts to the Center for Research on Computation and Society at Harvard University, both to support research on applied cryptography and society. Daniel J.~Moroz was also supported in part by the Ethereum Foundation.
\end{acks}

\bibliographystyle{ACM-Reference-Format}
\bibliography{refs}

\appendix

\section{Direct Calculation of Attack Feasibility}\label{app:analytic}

Using the state variable distributions and our feasibility function, we can directly calculate an approximation of the probability of a length-$n$ malicious reorgs by finding each permutation of the state space variables that corresponds to a feasible attack, and adding all of the probabilities of those occurring. Procedure \ref{proc:proc1} demonstrates this approach. 

However, there are two challenges. First, the domains of $A$ and $H$ are  infinite, thus we cannot sum over  every permutation of the state variables, and we have to approximate this by truncating. We chose a truncation point by using the criterion that we only consider states where every slot of the state occurs with probability greater than $10^{-8}$. This results in a very minimal underestimation of the true probability (on the order of $10^{-6}$ in our experiments), which is several orders of magnitude smaller than the actual probabilities as shown in Table \ref{tab:exact} in Appendix \ref{app:analyticTable}.  
\begin{lemma}\label{lem:153}
If $a_i \geq 153$, then $\mathrm{Pr}[\mathcal{S}_i] \leq 10^{-8}.$
\end{lemma}
\begin{proof}
We use $a_i$ as the limiting bound because we always assume that the attacker has less than 50\% of the total stake in the system. Thus the domain of $A$ will be the largest. First note that the probability of a slot $\mathcal{S}_i$ is less than the probability of each state variable in that slot, and namely,
\begin{align}
    \mathrm{Pr}[\mathcal{S}_i] \leq \mathrm{Pr}[a_i].
\end{align}
Thus we need to choose an upper bound, $a'$ for the domain of $A$ such that for all $a_i \geq a'$, $\mathrm{Pr}[a_i] \leq 10^{-8}$. The smallest $\alpha$ we consider is $0.1$, so using the definition of the geometric random variable we establish,
\begin{align}
    \Pro [a_i] = (0.9)^{a_i} \cdot 0.1.
\end{align}
This gives the following inequality,
\begin{align}
    \Pro [a_i] \leq 10^{-8} \implies  (0.9)^{a_i} \cdot 0.1 &\leq 10^{-8},
\end{align}
which we solve directly
\begin{align}
    a_i &\geq \frac{\ln{10^{-7}}}{\ln{0.9}} \geq 152.
\end{align}
\end{proof}

Second, the state space grows exponentially in the number of blocks in the attacking fork. For example, consider the case of a length-2 reorg, where we use the following ranges for our state variables,
$
    A = [0, 1, 2, ..., 152], 
    H = [0, 1, 2, ..., 152] , 
    E = [0, 1, 2, ..., 32].
$

Because we need three slots of state variables to compute the probability of a have a length-2 reorg, the number of permutations is $153^3 \cdot 153^3 \cdot 33^3 \approx 4.6 \times 10^{17}.$
There are optimizations to reduce the size of the state space, but in the best case it remains exponential. In order to come to reasonable conclusions about arbitrary length reorgs, we use a Monte Carlo method (see Appendix \ref{app:mc}) as well as importance sampling (see Appendix \ref{app:impsample}) to approximate the probability of malicious reorgs of length 2 and greater.

\floatname{algorithm}{Procedure}
\begin{algorithm}[H] 
\caption{Find attack probability}
\label{proc:proc1}
\begin{algorithmic}[1]
\REQUIRE $\alpha$
\STATE \texttt{totalProb} $\leftarrow 0$
\FOR{$\mathcal{S} \in (\texttt{A} \times \texttt{H} \times \texttt{E})^n$}
\STATE $\mathbf{a}, \mathbf{h}, \mathbf{e} \leftarrow \mathcal{S}$
\IF{\texttt{XOR}$(\mathbf{a}, \mathbf{h})$ and $\F(\mathcal{S})$}
\STATE \texttt{totalProb+=} $ \mathrm{Pr}[\mathcal{S}]$  
\ENDIF
\ENDFOR
\RETURN{\texttt{totalProb}}
\end{algorithmic}
\end{algorithm}
This procedure checks every permutation of the possible values of the state variables $\{\mathbf{a}, \mathbf{h}, \mathbf{e}\}$ by taking the Cartesian product of the possible values for $A, H, E$, $n$ times. For each permutation, it first checks if the length-$n$ reorg is feasible, and if it is, adds the probability of it occurring to the total probability. We use \texttt{XOR} as a vector function that is applied element wise to two sequences and that maps any non-zero value to \texttt{true} (e.g., $\texttt{XOR}([1,0,3],[0,4,0]) = \texttt{true}$, while $\texttt{XOR}([1,0,3],[1,4,0]) = \texttt{false}$). Note that this isn't an exact probability because the domains of the geometric random variables $A$ and $H$ are infinite.

\section{Analytic and Approx Comparison}\label{app:analyticTable}
\begin{table}[H]
    \centering
    \begin{tabular}{|c|c|c|c|}
        \hline
        $\alpha$ & Analytic & Approx. & Error \\
        \hline
        \hline
        0.10 & 0.000142 & $0.000141\pm 7.4 \times 10^{-6}$ & $1.00 \times 10^{-6}$ \\
        0.15 & 0.001419 & $0.001406\pm 2.3 \times 10^{-5}$ & $1.29 \times 10^{-5}$ \\
        0.20 & 0.007789 & $0.007758\pm 5.4 \times 10^{-5}$ & $3.10 \times 10^{-5}$ \\
        0.25 & 0.029502 & $0.029561\pm 1.1 \times 10^{-4}$ & $5.90 \times 10^{-5}$ \\
        0.30 & 0.081157 & $0.081210\pm 1.6 \times 10^{-4}$ & $5.30 \times 10^{-5}$ \\
        0.35 & 0.176913 & $0.176781\pm 2.3 \times 10^{-4}$ & $1.32 \times 10^{-4}$ \\
        0.40 & 0.323585 & $0.323729\pm 2.9 \times 10^{-4}$ & $1.44 \times 10^{-4}$ \\
        0.45 & 0.504535 & $0.504405\pm 3.1 \times 10^{-4}$ & $1.30 \times 10^{-4}$ \\
        \hline
    \end{tabular}
    \caption{Comparing the results for feasibility of length 1 reorg with the Monte Carlo method algorithm to the analytic results calculated using the probability distributions for a range of values of $\alpha$. Each of the MC values was calculated using a sample size of $10^7$. The 99\% Clopper-Pearson interval is shown as the range of approximate values. We see small error for each of the approximations. Note that the error increases as the probability of the event increases because variance of a binomial random variable increases until $\alpha=0.5$.} 
    \label{tab:exact}
\end{table}

\section{Monte Carlo Method}\label{app:mc}
Let $\mathcal{S}$ denote the joint distribution on the random state variables, and $\vec{S}$ a  sample from $\mathcal{S}$, with
\begin{align}
    \vec{S} &= 
    \begin{bmatrix}
        a_{\ell+1} & ... & a_{\ell+n} & h_{\ell+1} & ... & h_{\ell+n} & e_{\ell+1} & ... & e_{\ell+n}
    \end{bmatrix}^{T}.
\end{align}
 
 A Monte Carlo estimate of the probability of a feasible attack, given $N$ samples, is 
\begin{align}
    \hat{p} &= \frac{1}{N} \sum_{i=1}^N \mathcal{F} (\vec{S}_i), \; \text{where } \vec{S}_i \sim \mathcal{S}.
\end{align}

This is just the number of states satisfying the feasibility property divided by the sample size (recall that $\mathcal{F}$ returns 1 if an attack is feasible and 0 otherwise). This is an unbiased estimator for the true probability $p$.

Algorithms \ref{sub:samplegeneration} \& \ref{alg:mc} demonstrate the efficient sample generation and general Monte Carlo algorithm used. Table \ref{tab:exact} in Appendix \ref{app:analyticTable} shows the comparison of the Monte Carlo approximation to the analytic values described previously and demonstrate that they produce very similar results (i.e., the Monte Carlo approximation has low error).
\subsection{Efficient Sample Generation}
The following algorithm is used to generate a valid sample for each Monte Carlo trial. 
\floatname{algorithm}{Subroutine}\label{sub:samplegeneration}
\begin{algorithm}[H] 
\caption{Get sample $a_i, h_i, e_i$}
\label{proc:proc2}
\begin{algorithmic}[1]
\REQUIRE $\alpha$
\STATE $h \leftarrow$ \texttt{randomSample}$(\textrm{Geometric}(1-\alpha))$
\IF{h = 0}
\STATE $a \leftarrow$ \texttt{randomSample}$(\textrm{Geometric}(\alpha)) + 1$
\ELSE 
\STATE $a \leftarrow 0$
\ENDIF
\STATE $e \leftarrow \texttt{randomSample}(\textrm{Binomial}(32,\alpha))$
\RETURN $(a, h, e)$
\end{algorithmic}
\end{algorithm}
One subtle feature of this is the assigning of $a$ in Line 3. Because $h$ already has the value of $0$ for this draw, we know that $a$ cannot be zero, so we add one to the result.

\subsection{Standard Monte Carlo Method}
We use a standard Monte Carlo algorithm to estimate the probability of a feasible length-$n$ reorg. 
\floatname{algorithm}{Algorithm}\label{alg:mc}
\begin{algorithm}[H] 
\caption{Estimate probability of feasible attack - Standard Monte Carlo}
\label{proc:proc3}
\begin{algorithmic}[1]
\REQUIRE $\alpha$, \texttt{attackLength}, \texttt{sampleSize} 
\STATE \texttt{feasibleCount} $\leftarrow 0$
\FOR{$i = 0, 1, ... , \texttt{sampleSize}-1$}
\STATE \texttt{state} $\leftarrow \{\}$
\FOR{$ j = 0 , 1, ... , \texttt{attackLength}-1$}
\STATE $a, h, e \leftarrow \texttt{generateSample}(\alpha)$ 
\STATE \texttt{state.add}$(a,h,e)$
\ENDFOR
\IF{$\F(\texttt{state})$}
\STATE \texttt{feasibleCount++}
\ENDIF
\ENDFOR
\RETURN \texttt{feasibleCount / sampleSize} 
\end{algorithmic}
\end{algorithm}

\section{Importance Sampling}\label{app:impsample}

Monte Carlo estimation has the undesirable property that in the case of rare events, $n$ must be extremely large to obtain a tight bound for the confidence intervals around the estimator for $p$. In particular, we are interested in obtaining a probability estimate for attacks with probability, $ p \approx 10^{-8}$ of being feasible. To improve sample efficiency we make use of \textit{Importance Sampling}. 
This makes use of a different probability mass function $q(\vec{S})$, the proposal distribution,
which has a  higher weight on states that have the feasibility property. Then using samples from the new probability mass function, $\vec{S}_1, \vec{S}_2, ..., \vec{S}_n \overset{i.i.d.}{\sim} q(\mathcal{S})$, an 
unbiased estimator for $p$ is,
\begin{align}
    \hat{p} &= \frac{1}{n} \sum_{i=1}^n \frac{\mathcal{F}(\vec{S}_i)p(\vec{S}_i)}{q(\vec{S}_i)}.
\end{align}

We construct the following estimator for the variance of our estimated mean:
\begin{align}
    \hat{\sigma}^2 = \frac{1}{n} \sum_{i=1}^n \left[\frac{\mathcal{F}(\vec{S}_i)p(\vec{S}_i)}{q(\vec{S}_i)} - \hat{p} \right] ^2.
\end{align}
Thus we can construct a 99\% confidence interval for the true value of $p$ as,
$\hat{p} \pm 2.58 \hat{\sigma} / \sqrt{n}$. We now discuss the practical considerations of implementing importance sampling in this setting, as well as the $q$ distributions we used.

\subsection{Practical considerations of Importance Sampling}\label{app:practicalimportance}

We need to design a proposal distribution $q(\vec{S})$ such that the variance of our estimator is small. 
For this, we use a different value of $\alpha$ to parameterize the state variable distributions. The distribution
on state variables at each block height are,
    $A \sim \mathrm{Geometric}(\alpha),
    H \sim \mathrm{Geometric}(1- \alpha), 
    E \sim \mathrm{Binomial}(32, \alpha).$

For the proposal distribution  we choose $\alpha_q > \alpha$. By increasing the value of $\alpha$, we are ensured that the probability of a feasible attack will rise, and thus importance sampling detects them at a higher rate than classical Monte Carlo. However, it is important that we do not increase $\alpha$ by too much because this makes the likelihood ratio very small, and can lead to  numerical precision errors. Because we are using a standard double floating point representation with 64 bits and a 53 bit mantissa, we can expect around 16 decimal places of accuracy. Thus, we keep the likelihood ratio  above $10^{-16}$ to provide stability. Choosing $\alpha_q$ is a nontrivial process, and for this reason we only used importance sampling when the variance of the  Monte Carlo estimator is very high. Through experimentation we found that the heuristic $\alpha_q = \alpha + 0.05$ was sufficient for attacks of length 20 and 35, and the heuristic  $\alpha_q = \alpha + 0.03$ for attacks of length 55 and 80. This requires a bit of balancing to ensure that attacks are much more probable with the value of $\alpha_q$, but the likelihood ratio (the ratio of probability for a state to occur in $p$ over that same state probability in $q$) isn't too small. We found a suitable range of $\alpha_q$ by incrementing it by $0.01$ until the likelihood ratio fell below $10^{-16}$, and using a value approximately in the middle of that range.

\subsection{Variance reduction using Importance Sampling}

\begin{figure}[h!]
    \centering
    \includegraphics[width=0.37\textwidth]{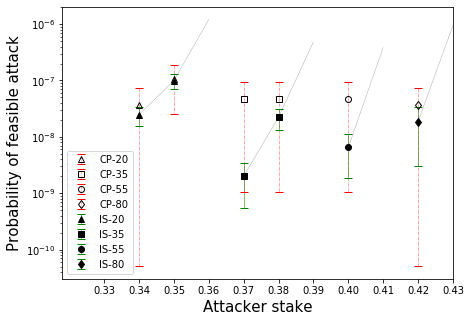}
    \caption{Importance sampling confidence intervals and Clopper-Pearson intervals for all the estimators with probability under $10^{-7}$. This demonstrates how effective importance sampling is at reducing the variance for our problem. Each marker is labelled with the reorg length and the color of the error bars denotes whether it is importance sampling (IS) or Clopper-Pearson (CP).}
    \label{fig:importancesamp}
\end{figure}

\end{document}